\documentclass[11pt]{llncs}

\usepackage{fullpage}
\usepackage{times}
\usepackage{amssymb}
\usepackage{amsmath}
\usepackage{color}
\usepackage{xspace}
\usepackage{cite}
\usepackage{subfigure}
\usepackage{epsfig}
\usepackage{psfrag}       
\usepackage[ruled,vlined,linesnumbered]{algorithm2e}
\usepackage{url}
\usepackage{listings}
\usepackage{ulem} 

\newcommand{\parhead}[1]{{\textbf{#1.}\xspace}}
\newcommand{\DN}{{Dynamic Network}\xspace}
\newcommand{\DNs}{{Dynamic Networks}\xspace}
\newcommand{\ADN}{{Anonymous Dynamic Network}\xspace}
\newcommand{\ADNs}{{Anonymous Dynamic Networks}\xspace}
\begin{document}

\title{A Faster Counting Protocol for \ADNs}

\titlerunning{A Faster Counting Protocol for \ADNs} 

\author{
   Alessia Milani\inst{1}
   \and
   Miguel~A.~Mosteiro\inst{2}
}
\authorrunning{A. Milani and M. A. Mosteiro}
\institute{
LABRI, University of Bordeaux, INP, France\\
\email{milani@labri.fr}
\and
Department of Computer Science, Kean University, USA\\ 
\email{mmosteir@kean.edu}
}

\maketitle              

\begin{abstract}
We study the problem of counting the number of nodes in a slotted-time communication network, under the challenging assumption that nodes do not have identifiers and the network topology changes frequently. That is, for each time slot links among nodes can change arbitrarily provided that the network is always connected.

This network model has been motivated by the ongoing development of new communication technologies that enable the deployment of a massive number of devices with highly dynamic connectivity patterns. Tolerating dynamic topologies is clearly crucial in face of mobility and unreliable communication. Current communication networks do have node identifiers though. Nevertheless, even if identifiers may be still available in future massive networks, it might be convenient to ignore them 
if neighboring nodes change all the time.  Consequently, knowing what is the cost of anonymity is of paramount importance to understand what is feasible or not for future generations of \DNs. 

Counting is a fundamental task in distributed computing since knowing the size of the system often facilitates the desing of solutions for more complex problems. Also, the size of the system is usually used to decide termination in distributed algorithms. 
Currently, the best upper bound proved on the running time to compute the exact network size is double-exponential. 
However, only linear complexity lower bounds are known, leaving open the question of whether efficient Counting protocols for \ADNs exist or not.

In this paper we make a significant step towards answering this question by presenting a distributed Counting protocol for \ADNs which has exponential time complexity. Our algorithm ensures that eventually every node knows the exact size of the system and stops executing the algorithm. Previous Counting protocols have 
either double-exponential time complexity, 
or they are exponential but do not terminate, 
or terminate but do not provide running-time guarantees, 
or guarantee only an exponential upper bound on the network size. 
Other protocols are heuristic and do not guarantee the correct count. 
\end{abstract}



\section{Introduction}
\label{sec:intro}

We study the problem of \emph{Counting} the number of nodes in a communication network, under the challenging assumption that nodes do not have identifiers (IDs) and the network topology changes frequently. 
We consider broadcast networks in slotted-time scenarios. That is, in any given time slot, a message sent by a given node is received by all nodes directly connected to it (\emph{one-hop neighbors}). 
Worst-case topology changes are modeled assuming the presence of an adversary that, for each time slot, chooses the set of links among nodes. The choice is arbitrary as long as, in each time slot, the network is connected. 
This dynamic topology model, called \emph{1-interval connectivity}, was introduced in~\cite{KuhnLO2010} for \DNs where each node has a unique identifier.

The network model described, called \emph{\ADN}, has attracted a lot of attention recently~\cite{spirakis,conscious,oracle,experimentalConscious}. The model has been motivated by the ongoing development of new communication technologies that enable the deployment of a massive number of devices with highly dynamic connectivity patterns. Tolerating dynamic topologies is clearly crucial in face of mobility and unreliable communication. 
Current communication networks do have node IDs (or otherwise a labeling is defined at startup). Nevertheless, in future massive networks, it might be suitable to avoid nodes IDs to facilitate mass production. Or, even if IDs are still available it might be convenient to ignore them under highly dynamic conditions where neighboring nodes change all the time.
Consequently, knowing what is the cost of anonymity is of paramount importance to understand what is feasible or not for future generations of \DNs. 

Counting is a fundamental distributed computing problem since knowing the size of the system facilitates the solution of more complex problems. Also this parameter is usually used to ensure the termination of the algorithm.

Counting can be solved in \ADNs, but the best known upper bound on the time complexity 
is double-exponential~\cite{conscious}. A double-exponential running time precludes the application of such algorithm to networks of significant size,
but only linear lower bounds are known.
Such a large gap leaves open the question of whether practical protocols exist or not.

The protocol presented in this paper makes a significant step towards answering 
the latter question,
reducing the time complexity for \emph{exact} Counting to exponential. 
Our algorithm ensures that there is a time slot when all nodes know the \emph{exact size} of the system and they stop executing the algorithm. All nodes stop at the same round and this is known by every node. Thus it is easy to concatenate another algorithm which uses the system size.

Previous Counting protocols for \ADNs have either double-exponential time complexity~\cite{conscious}, 
or they are exponential but do not terminate~\cite{conscious}, 
or terminate but do not provide running-time guarantees~\cite{oracle},
or guarantee only an exponential upper bound on the network size~\cite{spirakis}. 
Other protocols are heuristic and do not guarantee the correct count~\cite{experimentalConscious}.

All current Counting protocols for \ADNs~\cite{oracle,spirakis,conscious,experimentalConscious} assume the presence of one distinguished node, usually called \emph{leader}, and additionally use some knowledge of the number of neighbors of each node, called \emph{degree}. In our model, we include both assumptions. Namely, the presence of a leader node, and an upper bound on the maximum degree of the adversarial topology which is known by all nodes.
While these assumptions may seem too strong it was proved in~\cite{spirakis} that Counting is not solvable in Anonymous Networks without the presence of a \emph{leader}, even if the topology does not change. In the same work, it was conjectured that any non-trivial computation is impossible without knowledge of some network characteristics.


Our algorithm is inspired by the algorithm presented by Di Luna et al. in~\cite{conscious}, which starts computing an upper bound on the network size using the algorithm presented in~\cite{spirakis}.
Then, it verifies each candidate size down to the correct size. To verify each candidate size, an energy-transfer approach is used. Namely, each non-leader node is initially assigned a unit of energy which is shared evenly with neighbors in each communication round, except for the leader that works as a sink. This energy-transfer protocol is a backwards version of \emph{mass-distribution} and \emph{gossip-based} algorithms~\cite{KDGgossip,flowupdate,FMT:aggJournal} used to compute the size in other network models. The unit mass initially held in only one node in the latter system is shared throughout the network, converging to the average which is the inverse of the size. The energy-transfer protocol is shown to be at most exponential in the candidate size which in turn is exponential in the worst case, yielding a double-exponential Counting protocol.

The protocol presented here leverages the above idea of verifying candidate sizes using an energy-transfer protocol, but rather than starting with an upper bound, it follows a bottom-up approach. That is, it verifies $2,3,\dots,$ etc. up to the actual size. A carefully chosen energy threshold to decide when the count is accurate yields an exponential speedup in the worst-case running-time guarantees. 
The running time proved also identifies the collection of energy at the leader as the speedup bottleneck for gossip-based Counting, given that all other factors in the time complexity obtained are polynomial. In contrast, in the running time of other exact Counting protocols that terminate, all factors are exponential or double exponential~\cite{conscious}, or the running time is not proved~\cite{oracle}.

\subsection*{Contributions}
\label{sec:results}
In the following we summarize the main contributions of our work.
\begin{itemize}
\item Following-up on the Conscious Counting protocol of~\cite{conscious}, we present an improved Counting protocol for \ADNs that computes the exact number of nodes in less than $(2\Delta)^{n+1}(n+1)\ln(n+1)/\ln(2\Delta)$ communication rounds, where $n$ is the number of nodes and $\Delta$ is any upper bound on the maximum number of neighbors that any node will ever have. Our algorithm tolerates worst-case changes of topology, limited to 1-interval connectivity. The protocol requires the presence of one leader node and knowledge of $\Delta$. 

\item The running time of our protocol entails an exponential speedup over the previous best Counting algorithm in~\cite{conscious}, which was proved to run in $O(e^{(\Delta^{2n})}\Delta^{3n})$ communication rounds, which is double-exponential. 
The speedup attained is mainly due to a carefully chosen energy threshold used to verify candidate sizes that are not bigger than the actual size.  
Our analysis shows the correctness of such verification.

\item The time complexity proved identifies the phase where the leader collects energy from all other nodes as the speedup bottleneck for Counting with gossip-based protocols. Indeed, the exponential cost is due to this collection, whereas all other terms in the time complexity are polynomial. In contrast, in the running time of~\cite{conscious} all terms are exponential or double exponential.

\end{itemize}

\subsection*{Roadmap}
The rest of the paper is organized as follows. In Section~\ref{sec:related} we briefly overview previous work directly related to this paper. 
After formally defining the model and the problem in Section~\ref{sec:model}, we present our Counting protocol in Section~\ref{sec:protocols} and its analysis in Section~\ref{sec:analysis}. 


\section{Related Work}
\label{sec:related}

The following is an overview of previous work on Counting in \ADNs directly related to this paper. Other related work may be found in a survey on Dynamic Networks and Time-varying Graphs by Casteigts et al.~\cite{arnaudSurvey}, and in the papers cited below. 

Worst-case topology changes in \DNs may be limited assuming that the network is always connected (cf.~\cite{spirakis,KuhnLO2010,OW05, conscious}), or sometimes disconnected but for some limited time (cf.~\cite{geocast,PPC:opportunistic,F:delaytolerant,michail2014causality}).
The \emph{$T$-interval} connectivity model was introduced in~\cite{KuhnLO2010}. 
For $T\geq 1$, a network is said to be $T$-interval connected if for every $T$ consecutive rounds the network topology contains a stable connected subgraph spanning all nodes.
In the same paper, a Counting protocol was presented, but it requires each node to have a unique identifier. In~\cite{KuhnLO2010} it is also proved that, if no restriction on the size of the messages is required, the counting problem can be easily solved in $O(n)$ time when nodes have IDs. In our work, we focus on \emph{Anonymous} \DNs. Understanding if a linear counting algorithm exists also when IDs are not available will help to understand the difficulty introduced by anonymity (if any).

A Counting protocol for \ADNs where an upper bound $\Delta$ on the maximum degree is known was presented in~\cite{spirakis}. The adversarial topology is limited only to 1-interval connectivity, but the algorithm obtains only an upper bound on the size of the network $n$, which in the worst case is exponential, namely $O(\Delta^n)$. 
In our work, we aim to obtain an exact count, rather than only an upper bound.

The \emph{Conscious Counting} algorithm presented later in~\cite{conscious} does obtain the exact count for the same network model, but requires knowledge of an initial upper bound $K$ on the size of the network. Conscious Counting would be exponential if such upper bound were tight, since it runs in $O(e^{K^2}K^3)$ communication rounds. However, $K$ is obtained using the algorithm in~\cite{spirakis} mentioned above. Consequently,  in the worst case the overall running time of the Conscious Counting Algorithm is $O(e^{(\Delta^{2n})}\Delta^{3n})$, which is double-exponential. In our work, we obtain the exact count in exponential time. That is, we reduce exponentially the best known upper bound for exact Counting.

\ADNs where an upper bound on the maximum degree is not known where also studied~\cite{conscious,oracle,experimentalConscious}. In~\cite{conscious}, the protocol does not have a termination condition. That is, nodes running the protocol do not know whether the correct count has been reached or not. Hence, they have to continue running the protocol forever. In a companion paper~\cite{experimentalConscious}, the authors stop the protocol heuristically. Hence, the count obtained is not guaranteed to be correct. Indeed, errors appear when the conductance of the underlying connectivity graph is low.
In our work, we aim for Counting algorithms that terminate returning always the correct count. 
The protocol in~\cite{oracle} is shown to eventually terminate, although the running time is not proved. In their model, it is assumed that each node is equipped with an oracle that provides an estimation of its degree at each round. This is still an assumption of knowledge of network characteristics, although local.
This and the above shortcomings are not unexpected in light of the conjecture in~\cite{spirakis}, which states that
Counting (actually, any non-trivial computations) in \ADNs
without knowledge of some network characteristics is impossible.
Nevertheless, a proof of such conjecture has not been found yet.

Known lower bounds for Counting in \ADNs include only the trivial $\Omega(D)$, where $D$ is the \emph{dynamic} diameter of the network, and $\Omega(\log n)$~\footnote{Throughout the paper, $\log$ means logarithm base $2$, unless otherwise stated.} even if $D$ is constant, proved in~\cite{baldoni}.


\section{Preliminaries}
\label{sec:model}

\subsection{The Counting Problem} An algorithm is said to solve the \emph{Counting} problem if whenever it is executed in a \DN comprising $n$ nodes, all nodes eventually terminate and output $n$.

\subsection{The \ADN Model}
We consider a synchronous \DN composed of a fixed set of nodes $V$ where $|V|=n$. Nodes have no identifiers (IDs) or labels. We also assume the presence of a special node called the \emph{leader} and denoted $\ell$.  

Nodes communicate by broadcast. In particular, communication proceeds in synchronous \emph{rounds}. At each round a node broadcasts a message to its neighbors and simultaneously receives the messages broadcast in the same round by its neighbors (if any), then it makes some local computation. The time of computation is negligible. Thus, we compute the time complexity in rounds of communication.

At each round the set of communication links changes adversarially. Thus, the network is modeled as a dynamic graph $G=(V,E)$ where $E : \mathbb{N} \rightarrow \{(u,v) s.t. (u,v) \in V\}$ is a function mapping a round number $r$ to a set of undirected edges $E(r)$.  In particular, we consider the following $1$-interval connectivity model proposed by Kuhn et al. in \cite{KuhnLO2010}.

\begin{definition} A dynamic graph $G=(V,E)$ is 1-interval connected if for all $r\in\mathbb{N}$, the static graph $G_{r}:=(V,E(r))$ is connected.
\end{definition}
 
Finally, we assume that the size of the neighborhood of a node is upper bounded by a number $\Delta>0$ at every round, and we assume that $\Delta$ is known by the nodes.

At a first glance, some knowledge of the degree seems unnecessary because, after one message from each neighbor has been received in a given round, the degree is simply the message count. However, for the next round of communication, the degree may change due to changing topology. Thus, a node does not know its current degree before sending messages to its neighbors.

%
%
%
%
%
%
%
%
\section{Distributed Counting Algorithm}

\label{sec:protocols}
The algorithm consists of a sequence of iterations.
In each iteration, a candidate size is checked to decide if it is correct. 
If not, 
the candidate size is increased and a new iteration starts. 
In the following, we provide a high level explanation of the algorithm executed in each iteration.
 
At the beginning of each iteration every node is assigned energy value $1$, except for the leader which has $0$ energy. Then, the iteration proceeds in three consecutive phases described below. Each phase lasts a fixed amount of rounds which only depends on the current estimation of the system size. This is intended to synchronize the computation at all the nodes in the system without extra communication.

During the first phase, called the \emph{Collection Phase}, each node discharges itself by sending at each round a fraction at most half of its current energy to its neighbors. Then it computes its new energy by taking into account the energy given to its neighbors and the 
energy received from them.
The leader acts as a sink collecting energy but not disseminating it. This phase completes when the leader has received an amount of energy such that,
if the candidate size for the current iteration is the correct system size $n$, 
there is no node in the system with more than $1/k^c$ residual energy, for some constant $c>1$. The function $\tau(k)$ in Algorithms~\ref{algo:leader} and~\ref{algo:no-leader} gives the number of iterations of the Collection Phase needed to guarantee this. An exponential upper bound on $\tau(k)$ is computed in Corollary~\ref{cor:exp}. However, the bound may not be tight, so $\tau(k)$ is left as a parameter in the protocol. Should a better bound on $\tau(k)$ be proved, the protocol can be used as is.

Then, the \emph{Verification Phase} starts. During this phase, the energy at each node does not change and the leader verifies the correctness of 
the current candidate size
looking for a node with residual energy greater than $1/k^c$.  To this aim at each round of the Verification phase each non leader node broadcasts the maximal energy it has ``heard'' during this phase. At the beginning each such node broadcasts its own residual energy. This phase lasts sufficiently long to ensure that if a node with residual energy greater than $1/k^c$ exists, then the leader will hear from it. If the leader does not hear from 
such node, it knows that 
the candidate size was indeed correct, and the verification phase completes successfully.

The last phase, called \emph{Notification Phase}, is used by the leader 
when the verification phase completes successfully.
To notify such event, the leader 
broadcasts a special $\langle Halt\rangle$ message, 
and each node in turn broadcasts it as soon as 
it is received
and as long as the Notification Phase is not completed. 
If the Verification Phase completes unsuccessfully, the leader and every other node simply wait for the same number of rounds of communication
without taking any action, and then all the nodes start a new iteration. This procedure ensures synchronism.
A node stops executing the algorithm at the end of the Notification phase if it has received the $\langle Halt\rangle$ message. At this time every node knows the exact size of the system.

The protocol for the leader and non-leader nodes is detailed in Algorithms~\ref{algo:leader} and~\ref{algo:no-leader}.

\subsection*{PseudoCode}

\parhead{Variables at the leader node}
\begin{itemize}
\item $e_{\ell}$ is the energy of the leader at the current round. It is initialized to $0$ at the beginning of each iteration.
\item $k$ is the estimation of the system size. Initially equal to 1 and increased by one in each iteration.
\item $1/k^c$ is a threshold value for the energy such that, for a given estimate $k$, if $k$ is the correct size of the system, after the Collection Phase no node has energy greater than $1/k^c$ for some constant $c>1$.  
\item $IsCorrect$, initially $true$ is set to $false$ if the leader discovers that its estimate $k$ is wrong. This happens if the value of $e_{\ell}>k-1$ at the end of the Collection phase or if during the Verification phase the leader discovers a node with energy greater than $1/k^c$.
\item $halt$, initially $false$ is set to $true$ when the leader verifies that $k$ is the correct size of the system.
\end{itemize}

\begin{algorithm}[t]
\DontPrintSemicolon
$k\leftarrow 1$\;
$halt\leftarrow false$\;
\While{$\neg halt$}{
$k\leftarrow k+1$\;
$IsCorrect \leftarrow true$\;
 $e_{\ell}\leftarrow 0$\label{initzero}\;
\tcp{Collection Phase}
\For {each of $\tau(k)$ communication rounds }
 {\textsf{receive} $e_1,e_2,\ldots e_s$ \textsf{from neighbors, where $1\leq s\leq \Delta$}\;
   $e_{\ell}\leftarrow e_{\ell}+e_1+e_2+\ldots + e_s$\label{transferleader}\;
  }
\tcp{Verification Phase}
\For {each of $1+\left\lceil\frac{k}{1-1/k^c}\right\rceil$ communication rounds } { 
\textsf{receive} $e_1,e_2,\ldots e_s$ \textsf{from neighbors, where $1\leq s\leq \Delta$}\;
\eIf{$k-1-1/k^c\leq e_{\ell}\leq k-1$} 
{
\For {$j:= 1\ldots s$} {\If{$e_{j}>1/k^c$\label{detect}}{$IsCorrect \leftarrow false$\label{false1}\;}}
}
{$IsCorrect \leftarrow false$\label{false2}\;}
}
\tcp{Notification Phase}
\For {each of $k$ communication rounds } { 
\eIf{$IsCorrect$} {\textsf{broadcast} $\langle Halt\rangle$\;
$halt \leftarrow true$\;}
{\textsf{do nothing}\;}
}
}
output $k$\;
\caption{Algorithm of the leader node.}
\label{algo:leader} 
\end{algorithm}

\parhead{Variables at non leader nodes}
\begin{itemize}
\item $e$ is the energy of the node at the current round. It is initialized to $1$ at the beginning of each iteration.
\item $k$ is the estimation of the system size. Initially equal to 1 and increased by one in each iteration.
\item $e_{max}$, is the maximum energy the node is aware of at the current round of the Verification Phase.
\item $halt$, initially $false$, is set to $true$ when the node receives a $\langle Halt\rangle$ message.
\end{itemize}

\begin{algorithm}[t]
\SetKwFor{receive}{receive}{}{}
\DontPrintSemicolon
$k\leftarrow 1$\;
$halt\leftarrow false$\;
\While{$\neg halt$}{
$k\leftarrow k+1$\;
$e \leftarrow 1$\label{initone}\;
\tcp{Collection Phase}
\For {each of $\tau(k)$ communication rounds }
 {\textsf{broadcast} $\langle \frac{e}{2\Delta}\rangle $
  \textsf{and receive} $e_1,e_2,\ldots e_s$ \textsf{from neighbors, where $1\leq s\leq \Delta$}\;
   $e\leftarrow e\cdot (1-\frac{1}{2\Delta})+\sum^s_{j=1} e_j$\label{transfernonleader}\;
  }
\tcp{Verification Phase}
$e_{max}\leftarrow e$\label{aware}\;
\For {each of $\left\lceil1+\frac{k}{1-1/k^c}\right\rceil$ communication rounds } {
\textsf{broadcast} $\langle e_{max}\rangle $
\textsf{and receive} $e_1,e_2,\ldots e_s$ \textsf{from neighbors, where $1\leq s\leq \Delta$}\;
\For {$j:= 1\ldots s$}{\If{$e_{j}>e_{max}$}{$e_{max} \leftarrow e_{j}$\;}
}
}
\tcp{Notification Phase}
\For {each of $k$ communication rounds } { \If{halt} {\textsf{broadcast} $\langle Halt\rangle$\;}
\If{\textsf{\upshape receive} $\langle Halt\rangle$ \textsf{\upshape from some neighbor}\;}{$halt\leftarrow true$\;}
}
}
output $k$\;
\caption{Algorithm of non-leader nodes.}
\label{algo:no-leader} 
\end{algorithm}

\newpage

\section{Analysis}
\label{sec:analysis}

The following notation will be used. The energy of node $i$ at the beginning of round $r$, is denoted as $e_i^r$, which is also generalized to any set of nodes $S\subseteq V$ as $e_S^r=\sum_{i\in S}e_i^r$. 
For any given round $r$ and node $i$, 
let the set of neighbors of $i$ be $N_i^r$ and the average energy of $i$'s neighbors be $\overline{e}_{N_i^r}$. 
The superindex indicating the round number will be omitted when clear from context or irrelevant.
Also, at any time, let $\sum_{i\in V} e_i$ be called the \emph{system energy} and $\sum_{i\in V\setminus\{\ell\}} e_i$ be called the \emph{energy left}. 
At the beginning of each iteration of the protocol, that is, for each new size estimate $k$, the energy of the leader is reset to zero and the energy of the non-leader nodes is reset to $1$. Thus, the system energy is $\sum_{i\in V} e_i=n-1$ and the energy left is $\sum_{i\in V\setminus\{\ell\}} e_i=n-1$.

\begin{lemma}
\label{lemma:nodeenergy}
For any network of $n$ nodes, including a leader $\ell$, running the Counting Protocol under the communication and connectivity models defined the following holds.
For any given node $i\in V\setminus\{\ell\}$ and for any given round $r$ of the Collection Phase, it is $e_i^r\leq 1$.
\end{lemma}
\begin{proof}
Fix some arbitrary (non-leader) node $i$. 
Consider the transition between round $r$ and $r+1$. 
We have that
$$e_i^{r+1} \leq e_i^r + \overline{e}_{N_i^r} \frac{|N_i^r|}{2\Delta} - e_i^r \frac{|N_i^r|}{2\Delta} = e_i^r + (\overline{e}_{N_i^r} - e_i^r) \frac{|N_i^r|}{2\Delta}.$$

If $\overline{e}_{N_i^r} \leq e_i^r$, then $e_i^{r+1} \leq e_i^r$. 
That is, $i$'s energy does not increase from round $r$ to round $r+1$.
If on the other hand it is $\overline{e}_{N_i^r} > e_i^r$, we have
$$e_i^{r+1} \leq e_i^r + (\overline{e}_{N_i^r} - e_i^r)/2 = (e_i^r + \overline{e}_{N_i^r})/2.$$

That is, the energy of $i$ in round $r+1$ is at most the average between the energy of $i$ in round $r$ and the average of $i$'s neighbors' energy in round $r$. 

Now consider the evolution of the protocol along many rounds.
We ignore the rounds when $\overline{e}_{N_i^r} \leq e_i^r$ since they do not increase the energy. 
For the other rounds, given that all nodes start with energy $1$, and that the average of some numbers cannot be bigger than the maximum, the energy at any given node cannot get bigger than $1$. Hence, the claim follows. 
\end{proof}

\begin{lemma}
\label{lemma:broadcast}
For any network of $n$ nodes, under the communication and connectivity models defined, the following holds.
If a message $m$ is held by all nodes in a set $V_1\subseteq V$, after $|V|-|V_1|$ rounds when every node holding the message broadcasts $m$ in each round, all nodes in $V$ hold the message.
\end{lemma}
\begin{proof}
For any round $r>0$, consider the partition of nodes $\{V_1^r,V_2^r\}$ defined by the nodes holding the message at the beginning of round $r$. That is, $\forall i\in V_1^r$ the node $i$ holds $m$ and $\forall j\in V_2^r$ the node $j$ does not hold $m$. 
By 1-interval connectivity, there must exist a link $u,v$, such that $u\in V_1^r$ and $v\in V_2^r$. Given that all nodes holding the message broadcast $m$, $v$ must receive the message in round $r$. 
Thus, at the beginning of round $r+1$ it is $|V_1^{r+1}|\geq|V_1^r|+1$ and $|V_2^{r+1}|\leq|V_2^r|-1$. Applying the same argument inductively, after $|V_2^{r+1}|$ more rounds all nodes hold the message.
\end{proof}

The following lemma is a straightforward application of Lemma~\ref{lemma:broadcast} to the Notification Phase, where the message broadcasted is $\langle Halt\rangle$ for the first time when $k=n$.

\begin{lemma}
\label{lemma:halt}
{\bf Correctness of the Notification Phase:}
For any network of $n$ nodes, including a leader $\ell$, running the Counting Protocol under the communication and connectivity models defined the following holds.
If at the end of the Verification Phase $IsCorrect=true$, then at the end of the Notification Phase all nodes stop the Counting Protocol holding the size $n$.
\end{lemma}

\begin{lemma}
\label{lemma:correctness}
{\bf Correctness of the Verification Phase:}
For any network of $n>3$ nodes, including a leader $\ell$, running the Counting Protocol under the communication and connectivity models defined the following holds.
For any estimate of the size of the network $k$ and constant $c>1$, at the end of the Verification Phase $IsCorrect=true$ if and only if $k=n$.
\end{lemma}
\begin{proof}
We start observing that, for each estimate $k$, each non-leader node is initialized with one unit of energy (Line~\ref{initone} in Algorithm~\ref{algo:no-leader}) and the leader's energy is initialized to $0$ (Line~\ref{initzero} in Algorithm~\ref{algo:leader}). Until a new iteration of the outer loop (in both algorithms) is executed, no energy is lost or gained by the system as a whole. Hence, the system energy is always $n-1$.

We prove first that, if $k=n$, at the end of the Verification Phase it is $IsCorrect=true$.
Given that $k=n$, the system energy is $k-1$ and therefore $e_\ell\leq k-1$. Also because $k=n$, we know that after the Collection Phase it is $e_\ell\geq k-1-1/k^c$ by definition of $\tau(k)$.
Therefore, $IsCorrect$ is not set to false in Line~\ref{false2} of Algorithm~\ref{algo:leader}.
Also because $e_\ell\geq k-1-1/k^c$ at the end of the Collection Phase,
we know that the energy left at the beginning of the Verification Phase is 
$e_{V\setminus \{\ell\}} = k-1-e_\ell \leq 1/k^c$.
Therefore, no non-leader node could have more than that energy. 
That is, $\forall i\in V\setminus \{\ell\} : e_i\leq 1/k^c$. Thus, during the Verification Phase, the leader will not be able to detect a node with energy bigger than $1/k^c$.
Therefore, $IsCorrect$ is not set to false in Line~\ref{false1} of Algorithm~\ref{algo:leader} either.
There is no other line where $IsCorrect$ is set to false. Hence, at the end of the Verification Phase it is $IsCorrect=true$.

We prove now the other direction of the implication. That is, if at the end of the Verification Phase $IsCorrect=true$, then it is $k=n$.
For the sake of contradiction, assume that $IsCorrect=true$ but $k\neq n$.
Notice that $k$ cannot be larger than $n$, because the estimate is increased one by one, we already proved that if $k=n$ at the end of the Verification Phase it is $IsCorrect=true$, and Lemma~\ref{lemma:halt} shows that all nodes would have stopped running the protocol.
Thus, we are left with the case when $k<n$.

Notice that if $e_\ell > k-1$ the variable $IsCorrect$ is set to false in Line~\ref{false2} of Algorithm~\ref{algo:leader}.
Hence, it must be $e_\ell \leq k-1$ and,
given that the system energy is $n-1$, the energy left is $e_{V\setminus \{\ell\}}\geq n-k$.
This energy left is stored in the $n-1$ non-leader nodes. 
Hence, there must exist some node $j\in V\setminus\{\ell\}$ in the network such that $e_j\geq (n-k)/(n-1)$.
If $IsCorrect=true$ it means that the leader did not detect a node with energy bigger than $1/k^c$ in Line~\ref{detect} of Algorithm~\ref{algo:leader}. 
However, for any $2\leq k\leq n-1$, $n>3$, and $c>1$, it is $1/k^c<(n-k)/(n-1)$ which means that such node must exist. 

To see why the latter inequality is true, we verify that $k^c(n-k)-n+1>0$ as follows. With respect to $k$, this function has a maximum for $k=cn/(c+1)$. 
That is, for $2\leq k\leq n-1$ (recall that we are in the case $k<n$), the function has minima in $2$ and $n-1$. 
Then, it is enough to verify that $2^c(n-2)-n+1>0$, which is true for any $c>1$ and $n>3$, 
and that $(n-1)^c-n+1>0$, which is also true for any $c>1$ and $n\geq 2$.

Thus, to complete the proof, it is enough to show that $1+k^{c+1}/(k^c-1)$ rounds are enough to detect a node with energy bigger than $1/k^c$.
To do that, we upper bound the number of nodes in the network with energy at most $1/k^c$ as follows. We know that at any time when the leader has energy $e_\ell$, the energy left is $n-1-e_\ell$. Let $S\subseteq V$ be the set of nodes with energy at most $1/k^c$. Then, we have that $n-1-e_\ell=\sum_{j\in S}e_j+\sum_{k\in V\setminus S}e_k$. To maximize the size of $S$, we minimize the size of $V\setminus S$ assuming that all nodes in $V\setminus S$ have maximum energy, which according to Lemma~\ref{lemma:nodeenergy} is at most 1. Then, we have that $n-1-e_\ell=\sum_{j\in S}e_j+(n-|S|)$ which yields $|S|-1-e_\ell=\sum_{j\in S}e_j$ Given that $\sum_{j\in S}e_j\leq |S|/k^c$, we have that
$|S|\leq (1+e_\ell)/(1-1/k^c)$.
Recall that $e_\ell \leq k-1$ because $IsCorrect$ would have been set to false in Line~\ref{false2} of Algorithm~\ref{algo:leader} otherwise.
Replacing, we get $|S|\leq k^{c+1}/(k^c-1)$. 

Let $\{V_1,V_2\}$ be a partition of $V$ such that $V_2=S\cup\{\ell\}$. Recall that, for any $v\in V_1$ it is $e_v>1/k^c$. 
Using Lemma~\ref{lemma:broadcast}, we know that 
$|V_2|=1+k^{c+1}/(k^c-1)$ iterations in the Verification Phase of Algorithm~\ref{algo:leader} are enough for the leader to detect that there is a node with energy larger than $1/k^c$, which contradicts our assumption that $IsCorrect=true$. 
\end{proof}

The following theorem establishes our main result. 

\begin{theorem}
\label{thm:all}
For any anonymous dynamic network of $n>3$ nodes, including a leader $\ell$, and for any constant $c>\log 5$, the following holds.
If the adversarial topology is limited by a maximum degree $\Delta$ and the connectivity model defined, and nodes run the Counting Protocol in Algorithms~\ref{algo:leader} and~\ref{algo:no-leader} under the communication model defined, after $r$ rounds, all nodes stop holding the size of the network $n$, where
\begin{align*}
r 
&< n(n+3)+\ln n-4+\sum_{k=2}^n \tau(k).
\end{align*}
Where $\tau(k)$ is a function such that, if $k=n$ and the Collection Phase is executed for at least $\tau(k)$ rounds, then at the end of the phase the leader has energy $e_\ell\geq k-1-1/k^c$.
\end{theorem}
\begin{proof}
Correctness is a direct consequence of Lemmas~\ref{lemma:halt} and~\ref{lemma:correctness}. 
The running time is obtained adding the number of iterations of each phase, as follows.
\begin{align*}
r 
&= \sum_{k=2}^n \left(\tau(k)+\left\lceil1+\frac{k}{1-1/k^c}\right\rceil+k\right)\\
&\leq \sum_{k=2}^n \left(\tau(k)+2+\frac{k}{1-1/k^c}+k\right)\\
&= n(n+3)-4+\sum_{k=2}^n \left(\tau(k)+\frac{k}{k^c-1}\right).
\end{align*}

Using that $k/(k^c-1)<1/k$ for any $c>\log 5$ and $k\geq 2$, we obtain the following. 
\begin{align*}
r 
&< n(n+3)-4+\sum_{k=2}^n \left(\tau(k)+\frac{1}{k}\right)\\
&\leq n(n+3)+\ln n-4+\sum_{k=2}^n \tau(k).
\end{align*}
\end{proof}

Bounding the running time of the Collection Phase using Lemma 2 in~\cite{conscious} in Theorem~\ref{thm:all}, the following corollary is obtained.

\begin{corollary}
\label{cor:exp}
For any anonymous dynamic network of $n>6$ nodes, including a leader $\ell$, the following holds.
If the adversarial topology is limited by a maximum degree $1\leq\Delta\leq n-1$ and the connectivity model defined, and nodes run the Counting Protocol in Algorithms~\ref{algo:leader} and~\ref{algo:no-leader} under the communication model defined, after $r$ rounds, all nodes stop holding the size of the network $n$, where
\begin{align*}
r 
&< \frac{(2\Delta)^{n+1}(n+1)\ln (n+1)}{\ln(2\Delta)}.
\end{align*}
\end{corollary}
\begin{proof}
Lemma 2 in~\cite{conscious} proves that, for any estimate $k\geq n$ and integer $\rho>0$, starting with $e_\ell=0$ and $e_i=1$ for all $i\in V\setminus \{\ell\}$, after running $\rho k$ rounds of the energy transfer protocol the energy stored in the leader is $e_\ell \geq n(1-(((2\Delta)^{k}-1)/(2\Delta)^{k})^\rho)$. 
Notice in Theorem~\ref{thm:all} that the condition $e_\ell\geq k-1-1/k^c$ only applies when $k=n$.
Thus, it is enough to find $\rho$ such that 
\begin{align*}
k\left(1-\left(\frac{(2\Delta)^{k}-1}{(2\Delta)^{k}}\right)^{\rho}\right) &\geq k-1-1/k^c\\
\rho &\geq \frac{\ln (k/(1+1/k^c))}{\ln\left(1/(1-1/(2\Delta)^{k})\right)}.
\end{align*}

Using that $1-x\leq e^{-x}$ for $x\leq 1$, it is enough to have $\rho = \lceil(2\Delta)^k \ln k\rceil$. Replacing in Theorem~\ref{thm:all}, we obtain

\begin{align*}
r 
&< n(n+3)+\ln n-4+\sum_{k=2}^n k\lceil(2\Delta)^k \ln k\rceil\\
&\leq n(n+3)+\ln n-4+\sum_{k=2}^n k(1+(2\Delta)^k \ln k)\\
&= n(3n+7)/2+\ln n-5+\sum_{k=2}^n k(2\Delta)^k \ln k.
\end{align*}
 
Bounding with the integral,
\begin{align*}
r &< n(3n+7)/2+\ln n-5+\int_{k=2}^{n+1} k(2\Delta)^k \ln k~\mathrm{d} k\\
&= n(3n+7)/2+\ln n-5+\frac{(2\Delta)^k((k\ln(2\Delta)-1)\ln k-1)+\textrm{Ei}(k\ln(2\Delta))}{\ln^2(2\Delta)}\bigg|_2^{n+1}\\
&\leq n(3n+7)/2+\ln n+\frac{(2\Delta)^{n+1}(((n+1)\ln(2\Delta)-1)\ln (n+1)-1)+\textrm{Ei}((n+1)\ln(2\Delta))}{\ln^2(2\Delta)}.
\end{align*}


Using that $\textrm{Ei}(\ln x)=\textrm{li}(x)< x$, for any real number $x\neq 1$,
it is $\textrm{Ei}((n+1)\ln (2\Delta))<(2\Delta)^{n+1}$. Replacing,

\begin{align*}
r &< n(3n+7)/2+\ln n+\frac{(2\Delta)^{n+1}((n+1)\ln(2\Delta)-1)\ln (n+1)}{\ln^2(2\Delta)}\\
&= n(3n+7)/2+\ln n
+\frac{(2\Delta)^{n+1}(n+1)\ln (n+1)}{\ln(2\Delta)}
-\frac{(2\Delta)^{n+1}\ln (n+1)}{\ln^2(2\Delta)}.
\end{align*}
Using that $n(3n+7)/2+\ln n<(2\Delta)^{n+1}\ln (n+1)/\ln^2(2\Delta)$ for any $n>6$ and $1\leq\Delta\leq n-1$, the claim follows.
%
\end{proof}

\subsection{Discussion}
\label{sec:discussion}
In this paper we have studied the problem of Counting in \ADNs. The problem 
is challenging because the lack of identifiers and changing topology make difficult to decide if a new message has been received before from the same node. Also, the obvious lack of knowledge of the network size makes difficult to decide when the algorithm has to stop.

Assuming an upper bound on the size of the system facilitates termination but may lead to very bad time complexity if the upper bound is a huge overestimate. According to our knowledge, the algorithm in ~\cite{spirakis} is the only one to compute an upper bound of the system size for \ADNs and in the worst case it  is exponential, i.e. $O(\Delta^n)$ where $n$ is the size of the system and $\Delta$ is an upper bound on the nodes' degree. Finding the termination condition when an upper bound on the network size is not available is more challenging, but it is expected to provide more efficient algorithms. Our algorithm does not assume such upper bound, and computes the exact size of the system applying a bottom-up approach where the size is possibily underestimated several times.

It is known that if no restriction on the size of the messages is required, the counting problem can be easily solved in $O(n)$ time when nodes have IDs ~\cite{KuhnLO2010}). In this paper we have made a significant step towards understanding if a linear 
counting algorithm exists also when IDs are not available, by identifying the speedup bottleneck and reducing exponentially the best known upper bound.
This will help to understand the difficulty introduced by anonymity (if any). Despite our contribution, there is still a big gap 
with respect to the linear lower bound trivially given by the dynamic diameter.

Finally, although we focus on communication networks, our results carry over into any distributed system of similar characteristics. 

%
%
%
\section{Acknowledgements}
\label{section:ack}
We thank Arnaud Casteigts for introducing the model to us, and Antonio Fern\'andez Anta for useful discussions.
This study has been carried out with financial support from the French State, managed by the French National Research Agency (ANR) in the frame of the ``Investments for the future'' Programme IdEx Bordeaux - CPU (ANR-10-IDEX-03-02), 
from the National Science Foundation (CCF-1114930), 
and from Kean University UFRI Grant.


\bibliographystyle{abbrv}
\bibliography{./Comprehensive_2010}

\end{document}